\documentclass[11pt]{article}

\usepackage{amsmath,amssymb,amsfonts,amsthm}
\usepackage{graphicx}
\usepackage{subfigure}
\usepackage{fullpage}

\newtheorem{theorem}{Theorem}
\newtheorem{lemma}{Lemma}

\newtheorem{definition}{Definition}
\newtheorem{corollary}{Corollary}

\newcommand{\peel}[2]{\mathcal P(#1;#2)}

\newcommand{\dil}{\mathsf{dil}}
\newcommand{\len}{\ell}
\newcommand{\remove}[1]{}

\begin{document}

\title{Randomly removing  $g$ handles at once}


\author{Glencora Borradaile \\
Oregon State University
\and
James R.~Lee \\ University of Washington
\and
Anastasios Sidiropoulos \\
Toyota Technological Institute at Chicago}
%\fntext[label3]{}

\date{}

\maketitle

\begin{abstract}
  Indyk and Sidiropoulos (2007) proved that any orientable graph of
  genus $g$ can be probabilistically embedded into a graph of genus
  $g-1$ with constant distortion.  Viewing a graph of genus $g$ as
  embedded on the surface of a sphere with $g$ handles attached,
  Indyk and Sidiropoulos' method gives an embedding into a
  distribution over planar graphs with distortion $2^{O(g)}$, by
  iteratively removing the handles.  By removing
  all $g$ handles at once, we present a probabilistic embedding with
  distortion $O(g^2)$ for both orientable and non-orientable
  graphs. Our result is obtained by showing that the minimum-cut graph
  of Erickson and Har~Peled (2004) has low dilation, and then randomly
  cutting this graph out of the surface using the Peeling Lemma of Lee
  and Sidiropoulos (2009).
\end{abstract}

\section{Introduction}

Planar graphs constitute an important class of combinatorial
structures, since they can often be used to model
a wide variety of natural objects.
At the same time, they have
properties that give rise to improved algorithmic solutions for
numerous graph problems, if one restricts the set of possible inputs
to planar graphs (see, for example, applications to maximum
independent set \cite{Baker-planar} and computing maximum flows \cite{Borradaile-thesis}).

One natural generalization of planarity involves the genus of a graph.  Informally, a graph has genus $g$, for some
$g\geq 0$, if it can be drawn without any crossings on the surface of
a sphere with $g$ additional handles (see Section \ref{sec:prelims}).  For example, a planar graph has genus $0$, and a
graph that can be drawn on a torus has genus at most $1$.

In a way, the genus of a graph quantifies how far a graph is from
being planar.  Because of their similarities to planar graphs, graphs
of small genus usually exhibit nice algorithmic properties.  More
precisely, algorithms for planar graphs can usually be extended to
graphs of bounded genus, with a small loss in efficiency
or quality of the solution (e.g.~\cite{BDT09}). Unfortunately,
many such extensions are complicated and based on ad-hoc techniques.

Inspired by Bartal's probabilistic approximation of general metrics by
trees \cite{Bar96}, Sidiropoulos and Indyk~\cite{IS07} showed that every metric on
a graph of genus $g$ can be probabilistically approximated by a planar
graph metric with distortion at most exponential in $g$.
(See Section \ref{sec:prelims} for a formal definition of
probabilistic embeddings: randomized mappings between spaces
preserving distances in expectation.)  Since the distortion
measures the ability of the probabilistic mapping to preserve
metric properties of the original space, it is desirable to
make this quantity as small as possible.  In the present paper, we
show that the dependence of the distortion on the genus can be made
significantly smaller: $O(g^2)$ for graphs of orientable or
non-orientable genus $g$.  This requires a fundamental change over the
approach of \cite{IS07} which repeatedly reduces the genus by one
until the input graph is planar.

\medskip
\noindent {\bf Removing all the handles at once.}  Since (randomly)
removing handles one at a time incurs an exponential loss in
distortion, we look for a way to remove all the handles at once.

Our starting point is the {\em minimum-length cut graph} of Erickson
and Har-Peled \cite{EP-cutting}.  Given a graph $G$ the minimum-length
cut graph is (roughly speaking) a minimum-length subgraph $H$ of $G$
such that $G \setminus H$ is planar.  In Section \ref{sec:dilation},
we show that this $H$ is nearly {\em geodesically closed} in that
$d_H(u,v) \approx d_G(u,v)$ for all $u,v \in V(H)$, where $d_H$ and
$d_G$ are the shortest-path metrics on $H$ and $G$, respectively.
However, simply removing $H$ from $G$ could result in unbounded distortion for
some pairs of vertices of $G$. The geodesic-closure property suggests
that if we could {\em randomly shift} $H$, then the distortion of all
pairs of vertices in $G$ would be fine in expectation.

We use the {\em Peeling Lemma} of Lee and Sidiropoulos~\cite{LS08} to perform the random
shifting (Section~\ref{sec:peeling}).  The Peeling Lemma allows one to
randomly embed $G$ into a graph consisting of copies of $G \setminus
H$ hanging off an isomorphic copy of $H$, while keeping the expected
distortion of pairs of vertices in $G$ small.  The lemma requires an
appropriate random partition of the
shortest-path metric on $G$.  Such a procedure is provided
by the fundamental result of Klein, Plotkin, and Rao \cite{KPR93} for
partitioning graphs excluding a fixed minor.

This completes the proof, except that $H$ itself might not be planar.
However, $H$ does have small Euler number\footnote{The Euler
number we refer to exclusively in this paper is the value
$|E|-|V|+1$.}
and so admits a probabilistic embedding into a distribution over
trees~\cite{GNRS-journal,FRT-journal}.  In Section \ref{sec:embed}, we combine
these ingredients to provide a probabilistic embedding with
distortion $O(g^2$).

\medskip

In Section \ref{sec:lower}, we show that any such probabilistic
embedding incurs at least $\Omega(\log g)$ distortion. (A lower bound
of $\Omega(\log g/\log \log g)$ was given by Indyk and Sidiropoulos
in~\cite{IS07}.)  This still leaves an exponential gap between our
upper ($O(g^2)$) and lower ($\Omega(\log g)$) bounds.  We study the
limitations of our particular techniques and show an $\Omega(g)$ lower
bound for a restricted class of approaches.

\subsection{Preliminaries}
\label{sec:prelims}

Throughout the paper, we consider graphs $G=(V,E)$
with a non-negative length function $\len : E \to \mathbb R$.
We refer to these  as {\em metric graphs}. For pairs of vertices
$u,v \in V$, we
denote the length of the shortest path between $u$ and $v$ in $G$,
with the lengths of edges given by $\len$, by $d_G(u,v)$.
Unless otherwise stated, we restrict our attention
to finite graphs.

\paragraph{Graphs on surfaces}
Let us recall some notions from topological graph theory (an in-depth
exposition is provided by Mohar and Thomassen~\cite{MoharT-book}).  A \emph{surface} is a
compact connected 2-dimensional manifold, without boundary.
For a graph $G$ we can
define a one-dimensional simplicial complex $C$ associated with $G$ as
follows: The $0$-cells of $C$ are the vertices of $G$, and for each
edge $\{u,v\}$ of $G$, there is a $1$-cell in $C$ connecting $u$ and
$v$.  An \emph{embedding} of $G$ on a surface $S$ is a continuous injection
$f:C\rightarrow S$.  The \emph{orientable genus} of a graph $G$ is the smallest integer
$g\geq 0$ such that $C$ can be embedded into a sphere with $g$ handles.
The \emph{non-orientable genus} of $G$ is the smallest integer $k\geq 0$ such that $G$ can be embedded into a sphere with $k$ disjoint caps replaced by
disjoint M\"obius bands (also known as {\em cross caps}).
Note that a graph of genus $0$ is a planar graph.

\paragraph{Metric embeddings}
A mapping $F : X \to Y$ between two metric spaces $(X,d)$ and $(Y,d_Y)$
is {\em non-contracting} if $d_Y(F(x),F(y)) \geq d(x,y)$ for all $x,y \in X$.
If $(X,d)$ is any finite metric space, and $\mathcal Y$
is a family of finite metric spaces, we say that {\em $(X,d)$ admits
a stochastic $D$-embedding into $\mathcal Y$} if there exists
a random metric space $(Y,d_Y) \in \mathcal Y$ and a random
non-contracting mapping $F : X \to Y$ such that for every $x,y \in X$,
\begin{equation}
\label{eq:expansion}
\mathbb E\left[\vphantom{\bigoplus} d_Y(F(x),F(y))\right] \leq D \cdot d(x,y).
\end{equation}

The infimal constant $D$ such that \eqref{eq:expansion} holds is the {\em distortion of
the random mapping $F$.}
For two families $\mathcal F$ and $\mathcal G$ of metric graphs, we write $\mathcal F \rightsquigarrow \mathcal G$
if there exists an $D \geq 1$ such that every metric graph in $\mathcal F$ admits
a stochastic $D$-embedding into the family of metric graphs $\mathcal G$.  We will
write $\mathcal F \stackrel{D}{\rightsquigarrow} \mathcal G$ if we wish to emphasize the
particular constant.
A detailed exposition of  results
on metric embeddings can be found in \cite{I-survey} and \cite{Matousek-book}.

\section{Random planarization}

We show that every metric graph  of orientable or non-orientable genus $g$
embeds into a distribution over planar graph metrics with distortion at most $O(g^2)$.

\subsection{The peeling lemma}
\label{sec:peeling}

In this section, we review the Peeling Lemma from \cite{LS08}.
Let $G=(V,E)$ be a metric graph, and consider any subset $A \subseteq V$.
Let $G[A]$ denote the subgraph of $G$ induced by $A$, and let
$d_{G[A]}$ denote the induced shortest-path metric on $A$.
The {\em dilation of $A$ in $G$} is
$$
\dil_G(A) = \max_{x \neq y \in A} \frac{d_{G[A]}(x,y)}{d_G(x,y)}.
$$
Since $d_{G[A]}(x,y) \geq d_G(x,y)$ for all $x,y \in A$, $\dil_G(A) \geq 1$.

We now recall the following definition.
\begin{definition}[Lipschitz random partition]
\label{def:lrp}
For a partition $P$ of a set $X$, we write $P : X \to 2^X$ to denote the map
which sends $x$ to the set $P(x) \in P$ which contains $x$.
A partition $P$ of a metric space $X$ is $\Delta$-bounded
if every subset in $P$ has diameter at most $\Delta$.
A random partition $P$ of a
finite metric space $X$ is $\Delta$-bounded if the distribution
of $P$ is supported only on $\Delta$-bounded partition.

A random partition $P$ is $(\beta,\Delta)$-Lipschitz if it is $\Delta$-bounded, and
for every $x, y \in X$,
$$
\Pr[P(x) \neq P(y)] \leq \beta \frac{d(x,y)}{\Delta}.
$$
\end{definition}

For a metric space $(X,d)$, we write $\beta_{(X,d)}$ for the infimal $\beta$ such that $X$
admits a $(\beta,\Delta)$--Lipschitz random partition for every $\Delta > 0$,
and we refer to $\beta_{(X,d)}$ as the {\em decomposability modulus of $X$.}

%For any metric space $(X,d)$, one can associate a number
%$\beta_{(X,d)}$ called the {\em modulus of decomposability} of $(X,d)$
%(see~\cite{LS08} for a discussion).
The results of Rao \cite{Rao99} and Klein, Plotkin, and Rao \cite{KPR93}
yield the following for the special case of bounded-genus metrics.
(The stated quantitative dependence is due to \cite{FT03}; see also
\cite[Cor. 3.15]{LNinvent}.)

\begin{theorem}[KPR Decomposition]
\label{thm:kpr}
If $G=(V,E)$ is a metric graph of orientable or non-orientable genus $g \geq 0$,
then $\beta_{(V,d_G)} = O(g+1)$.
\end{theorem}

The dilation and modulus is used in the statement of the Peeling
Lemma. We use $G \stackrel{D}{\rightsquigarrow} H$ to denote the fact
that $G$ admits a stochastic $D$-embedding into the family $\{H\}$
(consisting only of the single graph $H$).

Let $G=(V,E)$ be any metric graph, and $A \subseteq V$.
We write $\peel{G}{A}$ for the graph defined as follows.
For each $a \in A$, we let $G_a$ be a disjoint isometric copy
of $G[(V \setminus A) \cup \{a\}]$.  Now, $\peel{G}{A}$
is the metric graph which arises by starting with the metric
graph $G[A]$ and identifying each node $a \in V(G[A])$ with
the corresponding node $a \in V(G_a)$.
Finally, for $a \in A$, we use $\bar a$ to denote
the corresponding copy of $a$ in $\peel{G}{A}$.

\begin{lemma}[Peeling Lemma, \cite{LS08}]
Let $G=(V,E)$ be a metric graph, and $A \subseteq V$ an arbitrary subset of vertices.
Let $G'=(V,E')$ be the metric graph with $E' = E \setminus E(G[A])$, and
let $\beta = \beta_{(V, d_{G'})}$ be the corresponding modulus of decomposability.
Then there is a stochastic $D$-embedding $F : V(G) \to V(\peel{G}{A})$, with
$D = O(\beta \cdot \dil_G(A))$.
Furthermore, for every $x,y \in A$, we have
\begin{equation}\label{eq:dilate}
d_{\peel{G}{A}}(F(x),F(y)) \leq \dil_G(A)\, d_G(x,y),
\end{equation}
and $F(a)=\bar a$ for all $a \in A$.
\end{lemma}

The idea of the Peeling Lemma is perhaps most easily illustrated
by the following model situation.  Let $G=(V,E)$ be a planar graph,
and consider an arbitrary subset $Q \subset V$ of the vertices.
One constructs the quotient graph $G/Q$ by identifying all the vertices
in $Q$ and taking the induced shortest-path metric.  Metrically,
the same
construction can be obtained by putting a clique on the vertices
of $Q$, and making all the edge lengths of the clique zero.
To formally apply the Peeling Lemma, we take the latter viewpoint,
so that $V(G)=V(G/Q)$.

Let $Q = \{q_1, q_2, \ldots, q_k\}$ and consider the graph $\peel{G}{Q}$
which, being a 1-sum of planar graphs, is itself planar.
The Peeling Lemma allows us to show that $G/Q \rightsquigarrow \peel{G}{Q}$, i.e.
the quotient metric on $G/Q$ can be probabilistically embedded
into a distribution over planar graphs (by applying it with $A=Q$).

Such an embedding of $G/Q$ into $Q$
can be described by a random mapping $f : V \to Q$ specifying
which copy of $G_{q_i}$ in $\peel{G}{Q}$ each vertex is mapped to,
with the constraint that $f(q)=q$ for every $q \in Q$,
i.e. $f|_Q$ is the identity.

In order for the random mapping to have low expansion
in expectation, it should be the case that nearby vertices $x,y \in V$
are often mapped to the same $G_{q_i}$, i.e. that $f(x)=f(y)$
with high probability.
This is where Lipschitz random partitions (Definition \ref{def:lrp})
come into play.  They allow us to partition the metric space
$(V,d_{G/Q})$ into disjoint pieces and then map all the vertices
in a single piece to the same layer of $\peel{G}{Q}$.  The Peeling Lemma is proved
by using such random partitions at many different scales (i.e.
many different values of the diameter parameter $\Delta$).
For more information, we refer to \cite{LS08}.

In the applications of the present paper, $G$ will be a graph of small genus,
and $A$ will be a set such that the removal of $G[A]$ leaves behind a planar graph.
Such sets are the topic of the next section.

\remove{
The following corollary is immediate from Theorem \ref{thm:kpr} and the Peeling Lemma, recalling
the fact that the family of planar graphs is closed under 1-sums.

\begin{corollary}\label{cor:peelgenus}
Let $G=(V,E)$ be a metric graph of orientable (resp., non-orientable) genus $g \geq 1$, and let $A \subseteq V$
be a subset of vertices such that the graphs $G[A]$ and $\displaystyle \left\{\vphantom{\bigoplus} G[V \setminus A \cup \{a\}]\right\}_{a \in A}$
are planar.  Then $G$ admits a stochastic $D$-embedding
into a distribution over planar graphs, with $D = O(g \cdot \dil_G(A))$.
\end{corollary}
}

%\input{non-separating-cycles.tex}

%Let $S_g$ denote the orientable surface of genus $g$.

\subsection{Low-dilation planarizing sets}
\label{sec:dilation}

In light of the Peeling Lemma, given a graph of bounded genus,
we would like to find a low-dilation set $A$ whose removal leaves behind
planar components.  In section~\ref{sec:embed}, we will deal with the fact
that the $G[A]$ might not be planar.  In everything that follows,
$\mathbb S$ will denote some compact surface of bounded (orientable or non-orientable)
genus.

\begin{definition}[Cut graph \cite{EP-cutting}]
Let $G$ be a graph embedded in $\mathbb S$.  Then, a subgraph $H$ of
$G$ is called a \emph{cut graph} if cutting $\mathbb S$ along the
image of $H$ results in a space homeomorphic to the disk.
\end{definition}

\noindent To {\em cut $\mathbb S$ along} an edge means to cut the surface (as
with a pair of scissors along that edge, creating two copies of it and
a hole.

\begin{definition}[One-sided walk]
Let $D$ be the disk obtained by a cut graph.  Every edge of $H$
appears twice in the boundary of $D$.  Let $x$ and $y$ be two
distinct vertices on the boundary of $D$. Let $R$ and $R'$ be the
paths bounding $D$ between $x$ and $y$.  An $x$-to-$y$ walk $X$ is
called {\em one-sided} if for every edge $e$ of $X$, $e$ is in $R$
and $e'$ (the copy of $e$) is in $R'$.
\end{definition}

\begin{lemma} \label{lem:one-sided-walks}
For any cut graph $H$ and any two vertices $x,y$ of $H$, there is a
one-sided walk from $x$ to $y$ in $H$.
\end{lemma}

\begin{proof}
Let $R$ and $R'$ be the $x$-to-$y$ boundaries of the disk resulting
from cutting the surface along $H$.  Let $C$ be the set of edges both of
whose copies are in $R$.  Let $A$ be the subgraph of $H$ containing
all the edges in $R$:  $A$ contains an $x$-to-$y$ path.  Let $B$ be
the subgraph of $H$ containing all the edges in $R\setminus C$.  We
prove $B$ contains an $x$-to-$y$ path.

Let $e$ be any edge in $C$ and let $e'$ be its copy.  Removing $e$
and $e'$ from $R$ is equivalent to glueing the disk along $e$ and
$e'$, creating a punctured surface.  Since edges of $R'$ are never
glued together, they will remain on the boundary of a common
puncture.  Let $S$ be the boundary of this puncture after glueing
all the edges of $C$ together.  Since $R'$ is an $x$-to-$y$ walk, $S
\setminus R'$ is an $x$-to-$y$ walk.  All the edges in $S\setminus
R'$ are in $R$ and their copies are in $R'$ by construction.
Therefore $S \setminus R'$ is a one-sided walk.
\end{proof}

Now we are ready to bound the dilation of minimum-length cut graphs.

\begin{lemma}\label{lemma:dilation-planarizing-upper}
Let $G$ be a graph embedded in $\mathbb S$.  If $H$ is a cut graph
of $G$ of minimum total length and $h$ is the number of vertices of
degree at least 3 in $H$, then $$\dil_G(V(H)) \leq h+2.$$
\end{lemma}

\begin{proof}
Assume for the sake of contradiction that there exist $x,y\in V(H)$,
with
$$d_H(x,y)> (h+2)\, d_G(x,y),$$ and pick $x,y$ so that $d_G(x,y)$ is
minimized among such pairs.  Let $Q$ be a shortest path between $x$
and $y$ in $G$.  Observe that by the choice of $x,y$, the path $Q$
intersects $H$ only at $x$ and $y$.  Let $D$ be the disk obtained
after cutting $\mathbb S$ along $H$.  Since $Q\cap H=\{x,y\}$, it
follows that the interior of $Q$ is contained in the interior of $D$
(Figure \ref{fig:cutgraph}).

Let $R,R'$ be the paths in $D$ between the end-points of $Q$,
obtained by traversing clockwise the boundary of $D$ starting from
$x$, and $y$ respectively (Figure \ref{fig:cutgraph}).  Let $S$ be
the one-sided walk that is contained in $R$, as guaranteed by
Lemma~\ref{lem:one-sided-walks}.  Let $J$ be any simple $x$-to-$y$
path contained in $S$.  We have
\begin{equation}\label{eq:J_Q}
  \len(J) \geq d_H(x,y) > (h+2)\, d_G(x,y) = (h+2) \len(Q)
\end{equation}

Let $U$ be the set of vertices of $H$ of degree at least $3$.  Since
$J$ is a simple path, it visits each vertex in $U$ at most once.  It
follows that the path $J$ consists of at most $k\leq h+1$ paths
$P_1,\ldots,P_k$ in $H$.  Since $J$ is a one-sided path, exactly one
copy of $P_i$ appears in $R$, for each $i\in [k]$.  Let
$j=\mbox{argmax}_{i\in [k]}\, \len(P_i)$.  By \eqref{eq:J_Q}, we
have
\begin{equation}\label{eq:Pj_small}
  \len(P_j) \geq \frac{\len(J)}{h+1} > \frac{h+2}{h+1} \len(Q) > \len(Q).
\end{equation}
Since $P_j$ has exactly one copy in each of $R$, and $R'$, it
follows that by cutting $D$ along $Q$ and gluing back along $P_j$ we
end up with a space homeomorphic to a disk (Figure
\ref{fig:cutgraph2}).  Therefore, the graph $H'$, obtained from $H$
by removing $P_j$ and by adding $Q$, is a cut graph.  By
\eqref{eq:Pj_small}, $H'$ has smaller total length than $H$,
contradicting the minimality of $H$, and concluding the proof.
\end{proof}

\begin{figure}
\begin{center}
\subfigure[The disk $\mathbb S \setminus \bar{H}$ with $R$, $R'$ on its boundary.]{
\scalebox{0.9}{\includegraphics{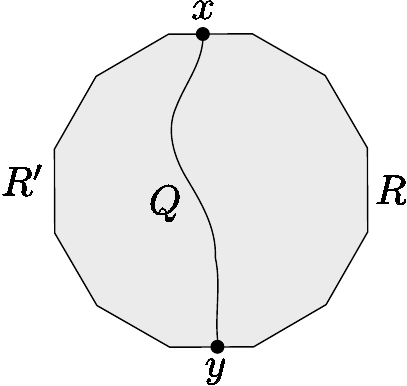}}
\label{fig:cutgraph}
}
\hspace{1cm}
\subfigure[Cutting along $Q$ and glueing along $P_j$.\label{fig:cutgraph2}]{
\scalebox{0.9}{\includegraphics{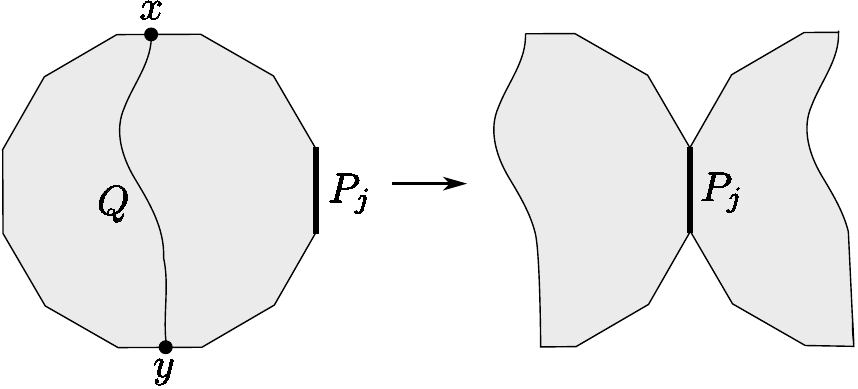}}
}
\caption{Building a smaller cut graph.}
\end{center}
\end{figure}

Now we state a result of \cite{EP-cutting}.

\begin{theorem}[Minimum-length cut graph,~\cite{EP-cutting}]\label{thm:EPcutting}
Let $G$ be a graph embedded on a surface $\mathbb S$ of orientable or non-orientable genus $g \geq 1$.
Then every minimal-length cut graph $H$ in $G$ has at most $4g-2$ vertices of degree at least 3.
Moreover, $H$ is the subdivision of some graph with at most $4g-2$ vertices and at most $6g-3$ edges.
\end{theorem}

\subsection{Applying the Peeling Lemma}
\label{sec:embed}

Given Lemma \ref{lemma:dilation-planarizing-upper} and Theorem \ref{thm:EPcutting}, we are in position to apply the Peeling Lemma,
except that the cut graph $H$ of Theorem \ref{thm:EPcutting} might not itself be planar.

\begin{theorem}
Let $G$ be any metric graph of orientable or non-orientable genus $g \geq 1$.  Then $G$
admits a stochastic $O(g^2)$-embedding into a distribution over planar graph metrics.
\end{theorem}

\begin{proof}
Let $G$ be embedded into a surface $\mathbb S$ of orientable or non-orientable genus $g \geq 1$.
Let $C$ be the cut graph given by Theorem \ref{thm:EPcutting}, which has at most $4g-2$ vertices
of degree at least 3.  Applying Lemma \ref{lemma:dilation-planarizing-upper}, we have
\begin{equation}\label{eq:dil}
\dil_G(V(C))\leq 4g.
\end{equation}
Since cutting $\mathbb S$ along $C$ gives a space homeomorphic to a
disk, it follows that $G\setminus C$ is planar.  $G \setminus C$ might
be disconnected, but this does not affect the argument.  See Figure
\ref{fig:dist1}.

Assume, without loss of generality,
that the minimum distance in $G$ is 1. (Since $G$ is finite, the distances
can always be rescaled to satisfy this constraint.)
Let $J$ be the graph obtained from $G$ as follows.
For every edge $\{u,v\}\in E(G)$, with $u\in V(C)$, and $v\notin V(C)$, we introduce a new vertex $z$, and we replace $\{u,v\}$ by a path $u$-$z$-$v$, with $d_J(z,v)=\frac12$, and $d_J(u,z)=d_G(u,v)-\frac12$.
Let $Z$ be the set of all these new vertices, and let $K=J[V(C)\cup Z]$, i.e. $K$ is the subgraph obtained from $C$ by adding all the new vertices in $Z$, and all the edges between $Z$ and $V(C)$.
Observe that for any $x,y\in V(G)$,
\[
d_J(x,y)=d_G(x,y).
\]
Therefore $G$ is isometric to $J$.
For any $x,y\in V(K)$, let $x',y'$ be the nearest neighbors of $x$ and $y$ in $V(C)$, respectively.
Since the minimum distance in $J$ is at least $\frac12$, we have
\begin{eqnarray*}
d_K(x,y) & = & d_K(x',y') + d_K(x,x') + d_K(y,y')\\
& \leq & 1 + d_C(x',y')\\
& \leq & 1 + 4g\, d_G(x',y')\mbox{ by Equation~\eqref{eq:dil}}\\
& = & 1+4g\, d_J(x',y') \\
& \leq & 1+4g\, (d_J(x,y) + d_J(x,x') + d_J(y,y')) \\
& \leq & 1+4g+4g\, d_J(x,y) \\
& \leq & (4g + 2(4g+1)) d_J(x,y) \\
& \leq & 14g\, d_J(x,y)
\end{eqnarray*}
% \begin{eqnarray*}
% d_K(x,y) & = & d_J(x',y') + d_J(x,x') + d_J(y,y')\\
%  & \leq & 1 + d_J(x',y')\\
%  & \leq & 2 d_J(x,y) + d_C(x',y')\\
%  & \leq & 2 d_J(x,y) + 4g \cdot d_G(x',y')\\
%  & = & 2 d_J(x,y) + 4g \cdot d_J(x',y')\\
%  & \leq & 2 d_J(x,y) + 4g (d_J(x,y) + d_J(x,x') + d_J(y,y'))\\
%  & \leq & 2 d_J(x,y) + 4g (d_J(x,y) + 1)\\
%  & \leq & 2 d_J(x,y) + 4g (d_J(x,y) + 2d_J(x,y))\\
%  & \leq & 14 g\, d_J(x,y)
% \end{eqnarray*}
Therefore,
\begin{equation}\label{eq:dil_K}
\dil_J(V(K)) \leq 14 g.
\end{equation}
Since $G$ is a graph of genus $g>0$, it follows that the modulus of decomposability of $J\setminus K$ is
\begin{equation}\label{eq:modulus_J_K}
\beta(J\setminus K)=\beta(G\setminus C)=O(g)
\end{equation}
by Theorem \ref{thm:kpr}.

Thus, by the Peeling Lemma and by (\ref{eq:dil_K}) and (\ref{eq:modulus_J_K}), we obtain that $J$ can be embedded into a distribution ${\cal F}$ over graphs obtained by 1-sums of $K$ with copies of $\{J[V(J)\setminus V(K) \cup\{a\}]\}_{a\in V(K)}$, with distortion at most $O(\beta(J\setminus K) \cdot \dil_J(V(K))) = O(g^2)$.
Observe that $J \setminus K = G\setminus C$, and thus $J\setminus K$ is a planar graph.
Moreover, for any $a\in V(K)$, there is at most one edge between $a$ and $J\setminus K$, and thus the graph $J[V(J)\setminus V(K) \cup\{a\}]$ is planar.
In other words, any graph in the support of ${\cal F}$ is obtained by 1-sums between $K$ and several planar graphs.

It remains to planarize $K$.  We observe that for pairs $x,y \in K$,
we have $$d_{G[K]}(x,y) \leq \dil_G(V(K))\,d_G(x,y) = O(g)\,d_G(x,y),$$ by
\eqref{eq:dilate} in the Peeling Lemma.
We  will embed $K$ into a random tree with distortion $O(\log g)$,
yielding an embedding of $G$ into planar graphs with total distortion
at most $O(g^2)$ (in fact, pairs in $K$ are stretched by only $O(g \log g)$ in expectation,
by the final assertion of the Peeling Lemma).

By Theorem \ref{thm:EPcutting}, the graph $C$ is the subdivision of a graph $C'$ with at most $4g-2$ vertices and at most $6g-3$ edges.
Recall that for a graph $\Gamma=(V,E)$, its Euler number is defined to be $\chi(\Gamma)=|E(\Gamma)| - |V(\Gamma)|+1$.
Clearly, the Euler number of a graph does not change by taking subdivisions, so we have
\begin{equation}\label{eq:euler}
\chi(C)=\chi(C') < |E(C')| - |V(C')| + 1 < 6 g.
\end{equation}
Gupta~{\em et~al.}~proved that any graph with Euler number
$\chi$ embeds into a dominating tree distribution with distortion
$O(\log \chi)$ (Theorem~5.5~\cite{GNRS-journal}).
Therefore by (\ref{eq:euler}) we obtain that $C$ can be embedded into a distribution ${\cal D}$ over dominating trees with distortion $O(\log{g})$.
Let $T$ be a random tree sampled from this distribution.
Let $T'$ be the graph obtained from $K$ by replacing the isometric copy of $C$ in $K$ by $T$.
Observe that every vertex $w\in V(K)\setminus V(C)$ is connected to a single vertex in $V(C)$, and has no other neighbors in $K$.
Therefore, the graph $T'$ is a tree with the same distortion as $T$.
It follows $K$ can also be embedded into a distribution ${\cal D}'$ over trees with distortion $O(\log g)$.

We are now ready to describe the embedding of $G$ into a random planar graph.
We embed $G$ into a random graph $W$ chosen from ${\cal F}$.
Recall that $W$ is obtained by 1-sums of a single copy of $J$ with multiple planar graphs.
Next, we embed $J$ into a random tree $Q$ chosen from ${\cal D}'$, and we replace the isometric copy of $J$ in $W$ by $Q$.
Let $R$ be the resulting graph, illustrated in Figure \ref{fig:dist2}.

\medskip

It is easy to see (by the triangle inequality and linearity of expectation), that the expansion of our random mapping
is worst for pairs $x,y \in V(G)$ with $(x,y) \in E(G)$.  For such an edge, if $x,y \in K$, we see
that the total expected stretch is $O(g \log g)$, while if one of $x$ or $y$ is not in $K$, then
the total expected stretch is at most $O(g^2)$ by our previous arguments.

This results in a stochastic $O(g^2)$-embedding of $G$ into $R$.
Moreover, the graph $R$ is the 1-sum of a tree with planar graphs.
Since the class of planar graphs is closed under 1-sums, it follows that $R$ is planar, concluding the proof.
\end{proof}

\begin{figure*}
\begin{center}
\subfigure[The graph $G$ with with the planarizing set $C$, and the graph $G\setminus C$.]{
\includegraphics{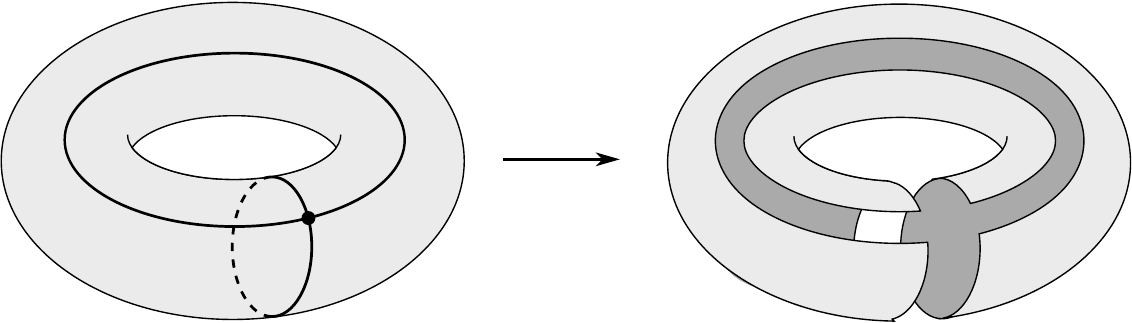}
\label{fig:dist1}
}
\hspace{1cm}
\subfigure[The 1-sum of $K$ with copies of $J\setminus K$, and the resulting graph obtained after replacing $K$ with a tree. Note that the figure is somewhat misleading as the tree is not necessarily a subtree of $K$.\label{fig:dist2}]{
\includegraphics{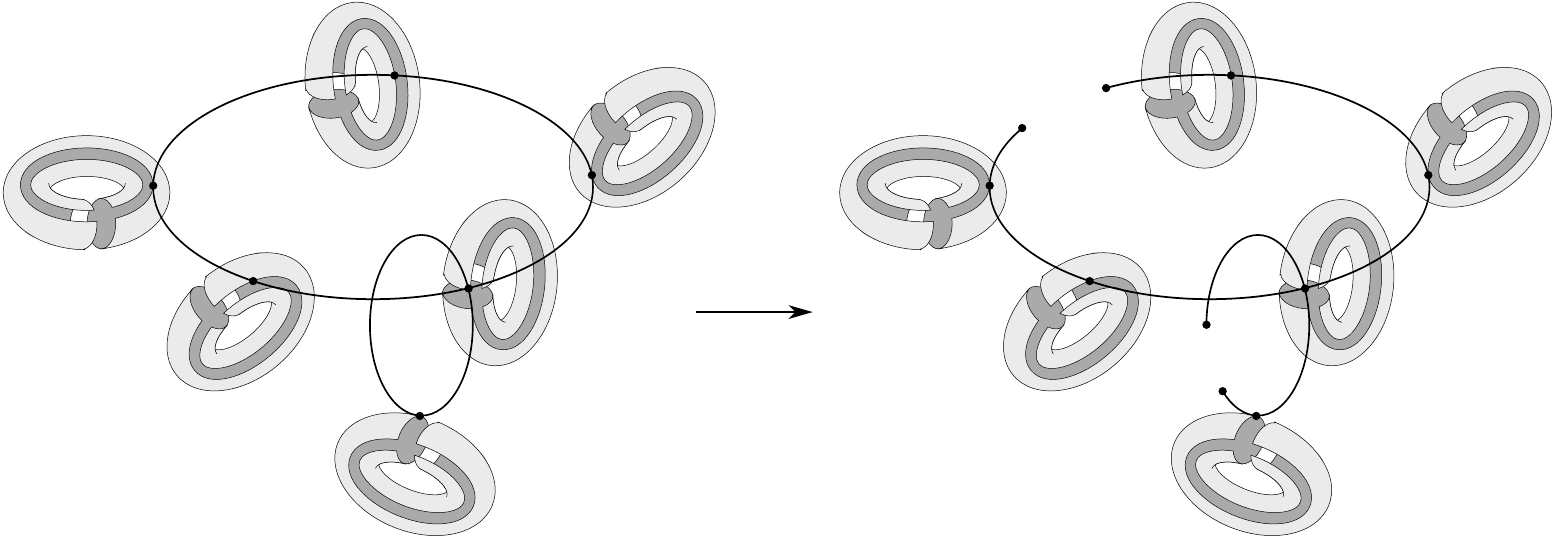}
}
\end{center}
\caption{Computing the stochastic embedding.}
\end{figure*}

\section{Lower bounds}
\label{sec:lower}

\subsection{The dilation of planar planarizing sets}
\label{sec:lowerdilate}
Given the exponential gap on the optimal distortion of a stochastic embedding of a genus-$g$ graph into a distribution over planar graphs ($O(g^2)$ vs. $\Omega(\log g)$), it is natural to ask whether the $O(g^2)$ bound on the distortion of our embedding is tight.
It is easy to construct examples of graphs of genus $g$ where a cut graph of minimum total length has dilation $\Omega(g)$.
In this case, our embedding clearly has distortion at least $\Omega(g)$, e.g. on the vertices of the planarizing set.

We can in fact show the following lower bound:
there are graphs of genus $g$ such that \emph{any} planar planarizing set has dilation $\Omega(g)$.
This implies that any algorithm that first computes a planar planarizing set $A$, and then outputs a stochastic embedding of 1-sums of $A$ with $G\setminus A$ using the Peeling Lemma, has distortion at least $\Omega(g)$.

\begin{theorem}
For any $g>0$, and for all sufficiently large $n$, there exists an $n$-vertex graph $G$ of genus $g$, such that for any planar subgraph $H$ of $G$, and $G\setminus H$ is planar, we have that the dilation of $H$ is at least $\Omega(g)$.
\end{theorem}

\begin{proof}
Let $S$ be a surface obtained from $K_5$ after replacing each vertex by a 3-dimensional sphere of radius 1, and each edge by a cylinder of length $n$ and radius 1.
Let also $S'$ be the surface obtained after attaching $g-1$ handles that are uniformly spread along $S$ (Figure \ref{fig:K5-tori}).
The diameter of each handle is $1$.
Clearly, the genus of $S'$ is $g$.
Observe that the minimum distance between any two such handles is $\Theta(n/g)$.
It is easy to see that we can triangulate $S'$ using triangles of
edge-length $\Theta(1)$, such that the set of vertices of the
triangulation is a $\Theta(1)$-net of $S'$ of size $n$.  Let $G$ be
the graph defined by this triangulation.

\begin{figure}[ht]
\begin{center}
\scalebox{0.5}{\includegraphics{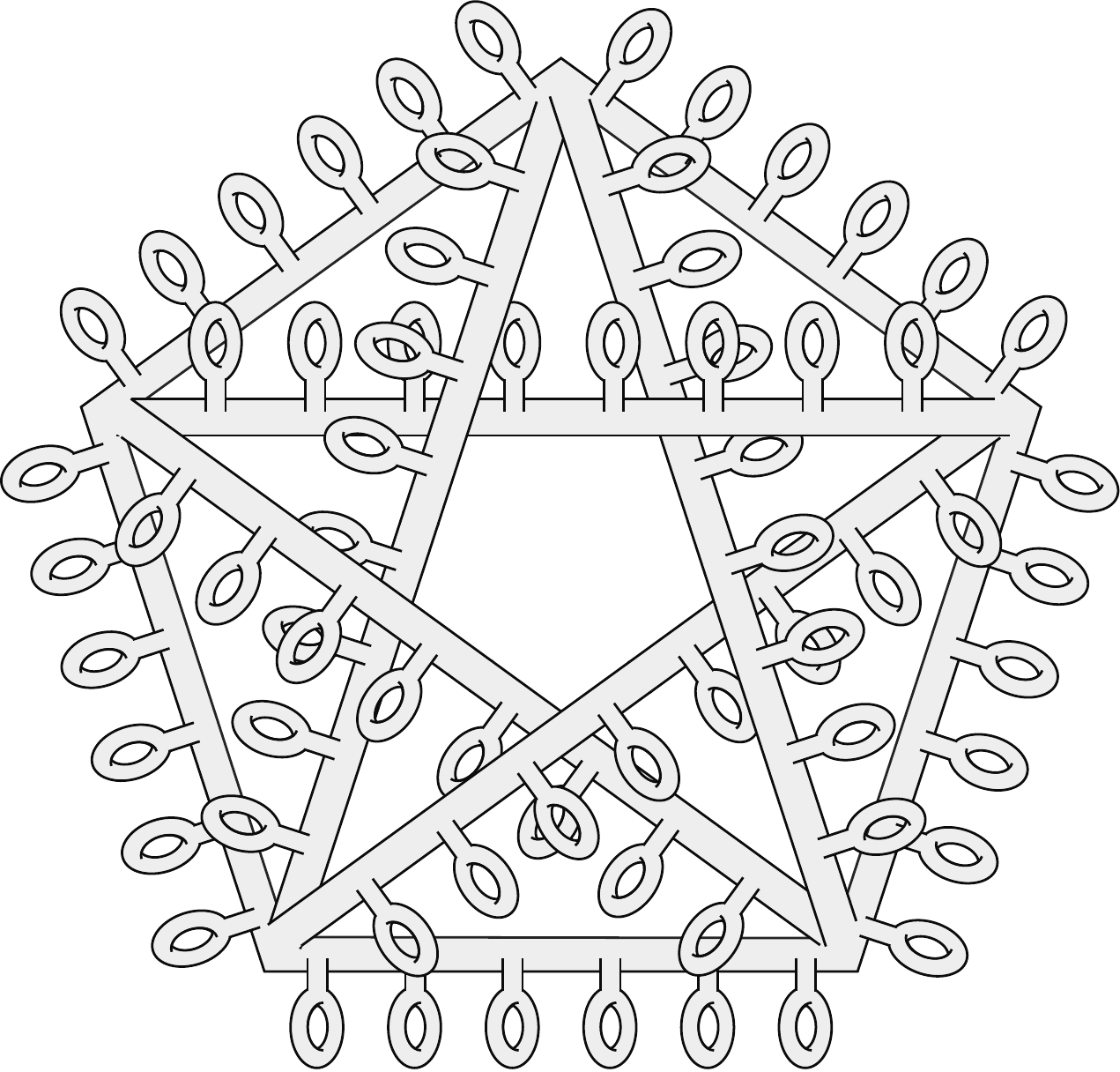}}
\caption{Tori on a surface corresponding to $K_5$.\label{fig:K5-tori}}
\end{center}
\end{figure}

Let now $H$ be a planar subgraph of $G$, such that $G\setminus H$ is also planar.
Since $G\setminus H$ is planar, it follows that $H$ contains at least one vertex in each handle that we added in $S$.
Let $J$ be the unweighted graph obtained after replacing each edge of a $K_5$ by a path of length $\frac{g}{100}$.
It follows that we can embed $J$ into $G$ with distortion $\Theta(1)$, such that the image of $V(J)$ is contained in $V(H)$.
Lemma 1 of \cite{Carroll-minor} states that any embedding of a graph $\Gamma_1$ containing a metric copy of the $k$-subdivision of a graph $\Gamma_2$, into a graph $\Gamma_3$ that excludes $\Gamma_2$ as a minor, has distortion $\Omega(k)$.
Therefore, any embedding of $J$ into a planar graph must have distortion $\Omega(g)$.
%This follows for example by Lemma 1 of \cite{Carroll-minor}.
Therefore, the dilation of $H$ is $\Omega(g)$.
\end{proof}

\subsection{Lower bounds for randomly planarizing a graph}

We now prove the following lower bound.

\begin{theorem}
For every $g \geq 1$, there exists a metric graph $G=(V,E)$ of orientable genus $O(g)$
such that if $G$ admits a stochastic $D$-embedding into a distribution
over planar graph metrics, then $D = \Omega(\log g)$.
\end{theorem}

\begin{proof}
If is easy to check that if a metric space $(X,d)$ admits a stochastic $D$-embedding $F$ into a family
$\mathcal Y$
of metric spaces, then the modulus of decomposability satisfies $\beta_{(X,d)} \leq D \cdot \sup_{(Y,d') \in \mathcal Y} \beta_{(Y,d')}$.
One simply takes a random $(\Delta,\beta)$-Lipschitz partition of $(Y,d')$, and pulls back the partition
to $X$ under $F$.  This yields a random partition of $X$ which is also $\Delta$-bounded, since
$F$ is non-contracting.  Observe that for $x,y \in X$, the probability they are separated
in the resulting partition is at most the expectation of $\beta \frac{d'(F(x),F(y))}{\Delta},$
which is at most $\beta D \frac{d(x,y)}{\Delta}$ by the property that $F$
is a stochastic $D$-embedding.

If $\mathcal Y$ is the family of planar graph metrics, then by Theorem \ref{thm:kpr}, we have $\sup_{(Y,d') \in \mathcal Y} \beta_{(Y,d')} = O(1)$, thus
$\beta_{(X,d)} = O(D)$.  On the other hand, there are $n$-point metric spaces (e.g. the shortest-path metric
on a constant-degree expander graph) which have $\beta_{(X,d)} = \Omega(\log n)$.  Combining this with the fact
that every $n$-point metric space can be represented as the shortest-path metric
of a graph with genus at most $O(n^2)$ yields the desired lower bound.
\end{proof}

\section{Open problems}
The most immediate problem left open by this work is closing the gap between the $O(g^2)$ upper bound, and the $\Omega(\log g)$ lower bound on the distortion.

Moreover, our embedding yields an algorithm with runnning time $n^{O(g)}$.
All the steps of the algorithm can be performed in time $n^{O(1)}\cdot g^{O(1)}$, except for the computation of the minimum-length cut graph, which is NP-complete, and for which the best-known algorithm has running time $n^{O(g)}$ \cite{EP-cutting}.
It remains an interesting open question whether the running time of our algorithm can be improved.
We remark that Erickson and Har-Peled \cite{EP-cutting} also give an algorithm for computing approximate cut graphs, with running time $O(g^2 n \log n)$.
These graphs however can have unbounded (in terms of $g$) dilation, so they don't seem to be applicable in our setting.

Another interesting open problem is whether there exist constant-distortion stochastic embeddings of bounded-genus graphs into planar \emph{subgraphs}.

\subsection*{Acknowledgements}

The authors are thankful for conversations with Jeff Erickson and
Sariel Har Peled.

\section*{Bibliography}

\bibliographystyle{elsarticle-num}
\bibliography{bibfile}

\end{document}